\documentclass{article}
\usepackage{amsmath}
\usepackage{amssymb}
\usepackage{amsfonts}
\usepackage{amsthm}
\usepackage{bm}
\usepackage[colorlinks,citecolor=blue]{hyperref}
\usepackage{color}
\usepackage{titlesec}
\usepackage{enumerate}

\titleformat{\section}{\Large\bfseries}{\thesection .}{0.4em}{}
\titleformat{\subsection}{\large\bfseries}{\thesubsection .}{0.4em}{}

\numberwithin{equation}{section}

\theoremstyle{plain}
\newtheorem{theorem}{Theorem}

\newtheorem{lemma}{Lemma}

\DeclareMathOperator*{\esssup}{ess\,sup}

\allowdisplaybreaks

\begin{document}

 
\title{Continuous Cluster Expansion for Field Theories}  
\author{Fang-Jie Zhao}  
\date{}      
\maketitle 

\begin{abstract}
\noindent
A new version of the cluster expansion is proposed without breaking the translation and rotation invariance. 
As an application of this technique, we construct the connected Schwinger functions 
of the regularized $\phi^4$ theory in a continuous way.

\end{abstract}

\textbf{keywords}
Constructive Field Theory; Cluster Expansion; Translation and Rotation Invariance; Connected Schwinger Function

\textbf{Mathematics Subject Classification}
81T08

\section{Introduction}

As a constructive tool for field theories, the cluster expansion \cite{FP07,BFP10} can be used to analyze  
the infinite volume limit in a rigorous manner. Contrary to the formal perturbative treatment, it is 
indeed convergent at least for weak coupling. Also its multiscale version, called phase cell expansion 
\cite{MS77,BF82}, plays the central role in the constructive renormalization.

However, the traditional cluster expansion breaks the translation and rotation invariance explicitly, 
since it is based on a fixed discretization of space-time \cite{GJ87,R91,AMR02} or a wavelet decomposition 
of the fields \cite{B99}. Although some remedies \cite{K86,B88,W89} have been invented, none of them 
is fully satisfactory and a manifestly Euclidean invariant expansion is desired. Thanks to the Pauli's 
principle, Fermionic field theories may be constructed continuously by some rearrangements and subtractions 
of the perturbative series \cite{AR98,DR00}. But continuous constructions of the Bosonic ones are more 
difficult, expect for some special cases such as $($Sine-Gordon$)_2$ with $\beta\!<\!4\pi$ \cite{BK87,BF93}. 
Recently, a new constructive technique for Bosonic field theories, called loop vertex expansion \cite{R07}, 
has been well developed. As one of its successes, the new technique has been used to construct the 
regularized $\phi^4$ theory in a continuous way \cite{MR08}. Unfortunately, when it is applied to the 
unregularized $\phi^4$ theory,  the cluster expansion is still required to remove the volume cut-off \cite{RW15}.

In this paper, we provide a continuous version of the cluster expansion, which is based on dynamical 
discretizations of space-time instead of a fixed one and thus retains the translation and rotation 
invariance explicitly. As an example, we construct the connected Schwinger functions of the regularized $\phi^4$ theory in this way. Generalizations to other stable interactions of polynomial type are straightforward. Moreover, 
the phase cell expansion may be realized in a similar way and then it seems promising to construct 
unregularized Bosonic field theories continuously. The progress in this direction will appear in future 
publications.

\section{Expansion}

Let us denote by $\mathcal{B}$ 
the collection of all Borel subsets of $\mathbb{R}^d$ and by $\mathcal{B}'$ the collection of all bounded ones in $\mathcal{B}$.
We consider a regularized $\phi^4_d$ theory with the interacting part restricted in 
$\Lambda\!\in\! \mathcal{B}'$. 
The generating functional of this theory is
\begin{align}
Z_{\Lambda}[J] = \!\int\! d\mu(\phi)\, e^{-\lambda\!\int_{\Lambda}\!dx\,\phi^4_x + \int_{\Lambda}\!dx\,J_x\phi_x},\;  
J \!\in\! \mathcal{S}(\mathbb{R}^d).
\end{align}
Here $d\mu$ is the Gaussian measure on the space of tempered distributions $\mathcal{S}'(\mathbb{R}^d)$ with covariance
\begin{align}\label{Covariance}
C_{x,y} = \frac{1}{(4\pi)^{d/2}} \int_{\frac12}^{1}\!\! \frac{d\alpha}{\alpha^{d/2}}\, e^{-\alpha-|x-y|^2\!/4\alpha},
\end{align}
which is a regularized version of the full covariance
\begin{align}
((-\Delta\!+\!1)^{-1})_{x,y} = \frac{1}{(4\pi)^{d/2}} \int_{0}^{\infty}\!\! \frac{d\alpha}{\alpha^{d/2}}\, e^{-\alpha-|x-y|^2\!/4\alpha}.
\end{align}
The connected Schwinger functions of the theory are
\begin{align}
S^c_{r,\Lambda;w_1,\dots,w_r} = \tfrac{\delta^r }{\delta J_{w_1}\!\cdots\delta J_{w_r}} \ln Z_{\Lambda}\big|_{J=0},\;\;
w_1,\dots,w_r\!\in\!\mathbb{R}^d.
\end{align}
(The $Z_2$ symmetry of the theory will not be used in the following, since the purpose of this paper is to provide a new expansion for general cases.) 
What we are really interested in is the limit theory as $\Lambda$ approaches $\mathbb{R}^d$,
which can be seen as a simplified model for a single slice in the renormalization group of 
the full $\phi^4_d$ theory \cite{FMRS87,BDH95,AR97}.

For $n \!\ge\! 0$, we denote by $\mathcal{I}_n$ the set of all sequences $(t_j)_{j=1}^n$ with $t_1,\dots,t_n \!\in\! [0,1]$ 
and by $\mathcal{X}_{n}$ the set of all sequences $(x_j)_{j=1}^n$ with $x_1,\dots,x_n\!\in\!\mathbb{R}^d$ 
satisfying $|x_j\!-\!x_{j'}|\!>\! 1$ for $j\!\neq\!j'$. 
For $\bm{x}\!\in\!\mathcal{X}_{0}$, let $B_{\bm{x}}\!=\!B'_{\bm{x}}\!=\!\varnothing$. 
For $(x_j)_{j=1}^n\!\in\!\mathcal{X}_{n}$ with $n \!\ge\! 1$, let
\begin{align}
B_{x_1,\dots,x_n} =  \big\{ y\!\in\!\mathbb{R}^d\!: \min\nolimits_{1\le j\le n}|y\!-\!x_j|\!\le\! 1 \big\}
\end{align}
and $B'_{x_1,\dots,x_n} = B_{x_1,\dots,x_n} \backslash B_{x_1,\dots,x_{n-1}}$.
We recursively define 
$C_{\bm{t};\bm{x}}$ for $\bm{t}\!\in\!\mathcal{I}_n$ and $\bm{x}\!\in\!\mathcal{X}_{n}$ as follows. If $n\!=\!0$, $C_{\bm{t};\bm{x}}\!=\!C$. Otherwise, 
\begin{align}\label{ConvexCombination}
C_{\bm{t};\bm{x}} = t_n C_{\bm{t}';\bm{x}'} + (1\!-\!t_n) \big( \chi_{B_{\bm{x}}} C_{\bm{t}';\bm{x}'} \chi_{B_{\bm{x}}}
\!+\! \chi_{(B_{\bm{x}})^{\mathsf{c}}} C \chi_{(B_{\bm{x}})^{\mathsf{c}}} \big)
\end{align}
for $\bm{t}\!=\!(\bm{t}',t_n)$ and $\bm{x}\!=\!(\bm{x}',x_n)$, where $\chi_A$ denotes the characteristic
function of subset $A$ of $\mathbb{R}^d$. 
By the convexity of the linear combination \eqref{ConvexCombination}, the interpolated covariance $C_{\bm{t};\bm{x}}$ 
remains positive  
and has the bound 
\begin{align}
0 \le (C_{\bm{t};\bm{x}})_{x,y} \le c_1 e^{-2|x-y|}
\end{align}
 from \eqref{Covariance}.
Let $d\mu_{\bm{t};\bm{x}}$ be the Gaussian measure on $\mathcal{S}'(\mathbb{R}^d)$ with covariance $C_{\bm{t};\bm{x}}$
and let $\langle F\rangle_{\bm{t};\bm{x}}^\phi$ be an abbreviation for $\int\! d\mu_{\bm{t};\bm{x}}(\phi)\, F[\phi]$.
For $\bm{t} \!=\! (t_j)_{j=1}^{n} \!\in\! \mathcal{I}_{n}$ and $\bm{x} \!=\! (x_j)_{j=1}^{n} \!\in\! \mathcal{X}_{n}$ with $n\!\ge\!1$,
\begin{align}
\frac{\partial C_{\bm{t};\bm{x}}}{\partial t_n}=
\sum_{k=1}^{n}\Big(\prod_{l=k}^{n-1}t_l\Big)
\Big(\chi_{B'_{x_1,\dots,x_k}}C\chi_{(B_{\bm{x}})^{\mathsf{c}}}+\chi_{(B_{\bm{x}})^{\mathsf{c}}}C\chi_{B'_{x_1,\dots,x_k}}\Big)
\end{align}
and then, by the formula for infinitesimal change of covariance \cite{GJ87},
\begin{align}\label{Partial_tn_F}
\frac{\partial}{\partial t_n}\langle F\rangle_{\bm{t};\bm{x}}^\phi
&= \!\int_{(B_{\bm{x}})^{\mathsf{c}}}\!\!\!dx \,\Big\langle \!\int_{B_{\bm{x}}}\!\!\!dy\, 
\Big( \frac{\partial C_{\bm{t};\bm{x}}}{\partial t_n}\Big)_{x,y} 
\frac{\delta}{\delta\phi_{x}} \frac{\delta}{\delta\phi_y} F \Big\rangle_{\bm{t};\bm{x}}^\phi \notag\\
&= \!\int_{(B_{\bm{x}})^{\mathsf{c}}}\!\!\!dx \,\Big\langle \Big( \sum_{k=1}^{n} \Big(\prod_{l=k}^{n-1}t_l\Big)
\Delta_{x;x_1,\dots,x_k}^\phi \Big) F \Big\rangle_{\bm{t};\bm{x}}^\phi,
\end{align}
where
$\Delta_{x;\bm{x}}^\phi$ is written simply for $\int_{B'_{\bm{x}}}\!dy\, \frac{\delta}{\delta\phi_x} C_{x,y} \frac{\delta}{\delta\phi_y}$. 

We define, for $\Lambda \!\in\! \mathcal{B}'$ and 
$(x_j)_{j=1}^n \!\in\! \mathcal{X}_{n}$ with $n\!\ge\! 1$, 
\begin{align}
\bm{\widetilde{Z}}_{\Lambda;x_1,\dots,x_n} &= \Big( \prod_{j=1}^{n-1} \!\int_0^1\!\!dt_j \Big)
\Big\langle \Big( \prod_{j=2}^{n} \sum_{k=1}^{j-1} \Big(\prod_{l=k}^{j-2}t_l \Big)\Delta_{x_{j};x_1,\dots,x_k}^\phi \Big) \notag\\
& \quad\;\; e^{-\lambda \!\int_{\Lambda}\!dx\,\phi^4_x + \int_{\Lambda}\!dx\,J_x\phi_x}
\Big\rangle_{(t_j)_{j=1}^{n-1};(x_j)_{j=1}^{n-1}}^\phi.
\end{align}
In particular,
$\bm{\widetilde{Z}}_{\Lambda;x_1} \!=\! Z_{\Lambda} \!:= \!\int\! d\mu(\phi)\, 
e^{-\lambda \!\int_{\Lambda}\!dx\, \phi^4_x + \int_\Lambda\!dx\,J_x\phi_x}$,
which is actually independent of $x_1$.
For $n\!\ge\! 2$, we call a map $\eta\!: \{2,3,...,n\}\!\to\!\{1,2,...,n\!-\!1\}$ with $\eta(j)\!<\!j$  an ordered tree with $n$ vertices and denote by $\mathcal{T}_n$ the set of all ordered trees with $n$ vertices.
Then $\bm{\widetilde{Z}}_{\Lambda;x_1,\dots,x_n}$ can be rewritten as
\begin{gather}
\sum_{\eta\in\mathcal{T}_n}\! 
\Big( \prod_{j=1}^{n-1} \!\int_0^1\!\!dt_j \Big) \Big( \prod_{j=2}^n \prod_{l=\eta(j)}^{j-2}t_l \Big) \notag\\
\Big\langle \Big( \prod_{j=2}^n \Delta_{x_j;x_1,\dots,x_{\eta(j)}}^\phi \Big) 
e^{-\lambda \!\int_{\Lambda}\!dx\,\phi^4_x + \int_{\Lambda}\!dx\,J_x\phi_x} \Big\rangle_{(t_j)_{j=1}^{n-1};(x_j)_{j=1}^{n-1}}^\phi.
\end{gather}
Write simply $\bm{Z}_{\Lambda;\bm{x}} \!=\! \bm{\widetilde{Z}}_{\Lambda \cap B_{\bm{x}};\bm{x}}$ for $\Lambda\!\in\!\mathcal{B}$.
By taking
\begin{align}
F[\phi] = \Big( \prod_{j=2}^{n} \sum_{k=1}^{j-1} \Big(\prod_{l=k}^{j-2}t_l\Big) \Delta_{x_{j};x_1,\dots,x_k}^\phi \Big) 
e^{-\lambda \!\int_{\Lambda}\!dx\,\phi^4_x + \int_{\Lambda}\!dx\,J_x\phi_x}
\end{align}
in \eqref{Partial_tn_F} and performing the integrations over $t_1,\dots,t_n$, we have
\begin{align}\label{identity_of_Z}
\bm{\widetilde{Z}}_{\Lambda;\bm{x}} = 
\bm{Z}_{\Lambda;\bm{x}} Z_{\Lambda\backslash B_{\bm{x}}} + 
\!\int_{ (B_{\bm{x}})^{\mathsf{c}}}\!\!\!dz\, \bm{\widetilde{Z}}_{\Lambda;\bm{x},z},
\end{align}
where the first summand factorizes since both of the Gaussian measure at $t_n\!=\!0$ and the integrand factorize.
Applying \eqref{identity_of_Z} successively, we obtain, for $z_1\!\not\in\!B_{\bm{x}}$,
\begin{align}\label{expansion_of_Z}
Z_{\Lambda\backslash B_{\bm{x}}} &=\bm{\widetilde{Z}}_{\Lambda\backslash B_{\bm{x}};z_1}=\bm{Z}_{\Lambda\backslash B_{\bm{x}};z_1}Z_{\Lambda\backslash B_{\bm{x},z_{1}}}+\!\int_{(B_{z_{1}})^{\mathsf{c}}}\!\!\!dz_{2}\,\bm{\widetilde{Z}}_{\Lambda\backslash B_{\bm{x}};z_{1},z_{2}}=\cdots\notag\\
&=\sum_{m\ge 1} \Big( \prod_{j=1}^{m-1} \int_{(B_{z_{1},\dots,z_{j}})^{\mathsf{c}}}\!\!\!dz_{j+1} \Big)\bm{Z}_{\Lambda\backslash B_{\bm{x}};z_{1},\dots,z_{m}} 
Z_{\Lambda\backslash B_{\bm{x},z_{1},\dots,z_{m}}},
\end{align}
where the summands vanish for $m$ sufficiently large by the boundedness of $\Lambda$
and the ranges of integrations $(B_{z_{1},\dots,z_{j}})^{\mathsf{c}}$ can be  replaced by $(B_{\bm{x},z_{1},\dots,z_{j}})^{\mathsf{c}}$.
Then, dividing both sides of \eqref{expansion_of_Z} by $Z_{\Lambda}$, we have the equation of Kirkwood-Salzburg type \cite{R69}
\begin{gather}
\frac{Z_{\Lambda\backslash B_{\bm{x},z_{1}}}}{Z_{\Lambda}} = \frac{Z_{\Lambda\backslash B_{\bm{x}}}}{Z_{\Lambda}}
- \big( \bm{Z}_{\Lambda\backslash B_{\bm{x}};z_{1}} \!-\! 1 \big)
\frac{Z_{\Lambda\backslash B_{\bm{x},z_{1}}}}{Z_{\Lambda}}  \,-\notag\\
\sum_{m\ge 2} \Big( \prod_{j=1}^{m-1} \int_{(B_{\bm{x},z_{1},\dots,z_{j}})^{\mathsf{c}}}\!\!\!dz_{j+1} \Big) \bm{Z}_{\Lambda\backslash B_{\bm{x}};z_{1},\dots,z_{m}} 
\frac{Z_{\Lambda\backslash B_{\bm{x},z_{1},\dots,z_{m}}}}{Z_{\Lambda}},
\end{gather}
which can be rewritten simply as
$\bm{f}_{\Lambda} \!=\! \bm{e} \!+\! \bm{A}_{\Lambda}\bm{f}_{\Lambda}$.
Here $\bm{e}$ and $\bm{f}_{\Lambda}$ are regarded as Borel measurable functions on $\bigcup_{n\ge 1}\mathcal{X}_n$
with $\bm{e}_{\bm{x}} \!\!=\! \bm{1}_{\bm{x}\in\mathcal{X}_1}$ and
$\bm{f}_{\Lambda;\bm{x}} \!=\! \frac{Z_{\Lambda\backslash B_{\bm{x}}}}{Z_{\Lambda}}$,
while $\bm{A}_{\Lambda}$ is regarded as a linear operator with
\begin{gather}
(\bm{A}_{\Lambda}\bm{f})_{\bm{x},z_1} = \bm{1}_{\bm{x}\not\in\mathcal{X}_0} \bm{f}_{\bm{x}} -
\big( \bm{Z}_{\Lambda\backslash B_{\bm{x}};z_{1}} \!-\! 1 \big) \bm{f}_{\bm{x},z_1}
\,- \notag\\
 \sum_{m\ge 2}\Big( \prod_{j=1}^{m-1} \int_{(B_{\bm{x},z_{1},\dots,z_{j}})^{\mathsf{c}}}\!\!\!dz_{j+1} \Big)
\bm{Z}_{\Lambda\backslash B_{\bm{x}};z_{1},\dots,z_{m}} \bm{f}_{\bm{x},z_{1},\dots,z_{m}}.
\end{gather}
(We say a function $\bm{f}$ on $\bigcup_{n\ge 1}\mathcal{X}_n$ is Borel measurable 
iff $\bm{f}|_{\mathcal{X}_n}$ is Borel measurable for each $n\!\ge\! 1$.)
Denoting
$\bm{Z}_{r,\Lambda;w_1,\dots,w_r} \!=\! \tfrac{\delta^r }{\delta J_{w_1}\!\cdots\delta J_{w_r}} 
\bm{Z}_{\Lambda}\big|_{J=0}$ and
$\bm{f}_{r,\Lambda;w_1,\dots,w_r} \!=\! \tfrac{\delta^r }{\delta J_{w_1}\!\cdots\delta J_{w_r}} \bm{f}_{\Lambda}\big|_{J=0}$,
we have
\begin{gather}
\big( \bm{1} \!-\! \bm{A}_{0,\Lambda} \big) \bm{f}_{r,\Lambda} = \bm{e}\, \delta_{r,0}
+ \sum_{0< s\le r} \bm{A}_{s,\Lambda} \bm{f}_{r-s,\Lambda}.
\end{gather}
Here $\bm{f}_{r,\Lambda}$ is regarded as a Borel measurable function on $\mathbb{R}^{d\times r} \times \bigcup_{n\ge 1}\mathcal{X}_n$,
while $\bm{A}_{s,\Lambda}$ is regarded as a linear operator 
with $\bm{A}_{0,\Lambda} \!=\! \bm{A}_{\Lambda}[0]$ and, for $0 \!<\! s \!\le\! r$,
\begin{gather}
( \bm{A}_{s,\Lambda} \bm{f} )_{w_1,\dots,w_r;\bm{x},z_1} = - 
\sum\nolimits_{\substack{I,I'\subset\{1,\dots,r\}\\|I|=s,I'=\{1,\dots,r\}\backslash I}} \sum_{m\ge 1} \notag\\
\Big( \prod_{j=1}^{m-1} \int_{(B_{\bm{x},z_{1},\dots,z_{j}})^{\mathsf{c}}}\!\!\!dz_{j+1} \Big)
\bm{Z}_{s,\Lambda\backslash B_{\bm{x}};(w_i)_{i\in I};z_{1},\dots,z_{m}} 
\bm{f}_{(w_i)_{i\in I'};\bm{x},z_{1},\dots,z_{m}}.
\end{gather}
Now letting $\bm{f}^*_{r}$ and $\bm{A}_{r}$ be the formal limits of $\bm{f}_{r,\Lambda}$ and $\bm{A}_{r,\Lambda}$ respectively as $\Lambda \!\to\! \mathbb{R}^d$,
we have that
\begin{gather}\label{equations_of_f}
\big( \bm{1} \!-\! \bm{A}_{0} \big) \bm{f}^*_{r} = \bm{e} \,\delta_{r,0} + \sum_{0<s\le r} \bm{A}_{s} \bm{f}^*_{r-s}
\end{gather}
with 
\begin{align}
( \bm{A}_{0} \bm{f} )_{\bm{x},z_1} = \bm{1}_{\bm{x}\not\in\mathcal{X}_0} \bm{f}_{\bm{x}} - 
\big( \bm{Z}_{0,(B_{\bm{x}})^{\mathsf{c}};z_{1}} \!-\! 1 \big) \bm{f}_{\bm{x},z_1}
- \sum_{m\ge 2} \notag\\
\Big( \prod_{j=1}^{m-1} \int_{(B_{\bm{x},z_{1},\dots,z_{j}})^{\mathsf{c}}}\!\!\!dz_{j+1} \Big)
\bm{Z}_{0,(B_{\bm{x}})^{\mathsf{c}};z_{1},\dots,z_{m}} \bm{f}_{\bm{x},z_{1},\dots,z_{m}}
\end{align}
and, for $0 \!<\! s \!\le\! r$,
\begin{gather}
( \bm{A}_{s} \bm{f} )_{w_1,\dots,w_r;\bm{x},z_1} = - 
\sum\nolimits_{\substack{I,I'\subset\{1,\dots,r\}\\|I|=s,I'=\{1,\dots,r\}\backslash I}} \sum_{m\ge 1} \notag\\
\Big( \prod_{j=1}^{m-1} \int_{(B_{\bm{x},z_{1},\dots,z_{j}})^{\mathsf{c}}}\!\!\!dz_{j+1} \Big)
\bm{Z}_{s,(B_{\bm{x}})^{\mathsf{c}};(w_i)_{i\in I};z_{1},\dots,z_{m}}
\bm{f}_{(w_i)_{i\in I'};\bm{x},z_{1},\dots,z_{m}}.
\end{gather}

To exhibit the exponential decay of the connected Schwinger functions, 
we review the following two tree-lengths for $x_1,\dots,x_n\!\in\!\mathbb{R}^d$ \cite{DIS73}.
One is
\begin{align}
\ell'_{x_1,\dots,x_n} = \min_{T}\sum_{(i,j)\in T} |x_i\!-\!x_j|,
\end{align}
where the minimum is taken over all trees $T$ connecting $x_1,\dots,x_n$. 
The other is
$\ell_{x_1,\dots,x_n} \!=\! \inf
 \ell'_{x_1,\dots,x_n, y_1,\dots,y_m}$,
where the infimum is taken over  all finite subsets $\{y_1,\dots,y_m\}$ of $\mathbb{R}^d$ (including $\varnothing$).
In particular, $\ell'_{x_1}\!=\!\ell_{x_1}\!=\!0$ and $\ell'_{x_1,x_2}\!=\!\ell_{x_1,x_2}\!=\!|x_1\!-\!x_2|$.
It is easy to see that $\ell'_{x_1,\dots,x_n}$ and $\ell_{x_1,\dots,x_n}$ are symmetric in $x_1,\dots,x_n$.
Also we have the inequality $\tfrac12\ell_{x_1,\dots,x_n}\le \ell'_{x_1,\dots,x_n} \le \ell_{x_1,\dots,x_n}$ \cite{DIS73},
indicating the equivalence of $\ell'_{x_1,\dots,x_n}$ and $\ell_{x_1,\dots,x_n}$ in some sense.
We can also define these two tree-lengths for nonempty subsets $X_1,\dots,X_n\!\subset\!\mathbb{R}^d$ as
\begin{align}
\ell'_{X_1,\dots,X_n} = \min_{T}\sum_{(i,j)\in T} \inf_{x\in X_i, x'\in X_j}|x\!-\!x'|
\end{align}
and $\ell_{X_1,\dots,X_n}\!=\! \inf
 \ell'_{X_1,\dots,X_n, \{y_1\},\dots,\{y_m\}}$ with the infimum also taken over all finite subsets $\{y_1,\dots,y_m\}$ of $\mathbb{R}^d$.
Writing simply $\ell_{x_1,\dots,x_n;y_1,\dots,y_m}\!=\!\ell_{\{x_1\},\dots,\{x_n\},\{y_1,\dots,y_m\}}$,
we list the following useful inequalities for $\ell$ and leave the proofs of them to the reader:
\begin{gather}
\ell_{x_1,\dots,x_n}\le\ell_{x_1,\dots,x_n,\dots,x_{n+n'}}\le \ell_{x_1,\dots,x_n}+\ell_{x_n,\dots,x_{n+n'}},
\label{inequalities1_of_ell}\\
\ell_{x_1,\dots,x_n;y_1,\dots,y_m}\ge\ell_{x_1,\dots,x_n;y_1,\dots,y_m,\dots,y_{m+m'}}
\ge\ell_{x_1,\dots,x_n;y_1,\dots,y_m}-\ell_{y_m,\dots,y_{m+m'}}.\label{inequalities2_of_ell}
\end{gather}

For $r\!\ge\! 0$, let $\mathcal{F}_{r}$ be the Banach space of 
Borel measurable functions on $\mathbb{R}^{dr} \times \bigcup_{n\ge 1}\mathcal{X}_n$
with the norm
\begin{align}
\|\bm{f}\|_{r} = \tfrac{1}{r!} \sup_{n\ge 1} 
\esssup_{\bm{w}\in \mathbb{R}^{dr},\bm{x}\in \mathcal{X}_n} 2^{1-n} 
e^{\ell_{\bm{w};\bm{x}}} |\bm{f}_{\bm{w};\bm{x}}|,
\end{align}
where $\ell_{\bm{w};\bm{x}}$ is used to show the exponential decay of $\bm{f}_{\bm{w};\bm{x}}$ in $\bm{w}$.
In particular, 
\begin{align}
\|\bm{f}\|_{0} = \sup_{n\ge 1} \esssup_{\bm{x}\in \mathcal{X}_n} 2^{1-n}|\bm{f}_{\bm{x}}|.
\end{align}
Also let
$\|\bm{A}\|_{s\to r} = \sup_{\|\bm{f}\|_{s}=1}
\|\bm{A}\bm{f}\|_{r}$
for linear  operator $\bm{A}\!: \mathcal{F}_{s} \!\to\! \mathcal{F}_{r}$.

\begin{theorem}
For $\lambda \!>\! 0$ sufficiently small, we have
$\|\bm{A}_{0}\|_{r\to r} \!\le\! \tfrac34$ and, for $0 \!<\! s \!\le\! r$,
$\|\bm{A}_{s}\|_{r-s\to r} \!\le\! c^s$.
\end{theorem}

By Theorem 1, \eqref{equations_of_f} has a unique solution $(\bm{f}^*_r)_{r\ge 0}$ with
\begin{align}
\bm{f}^*_{r} = \big( \bm{1} \!-\! \bm{A}_{0} \big)^{-1}\big(\bm{e} \,\delta_{r,0} + \sum_{0<s\le r} \bm{A}_{s} \bm{f}^*_{r-s}\big)\in\mathcal{F}_{r}
\end{align}
for $\lambda \!>\! 0$ sufficiently small.
Since $\|\bm{e}\|_{0} \!=\! 1$, we have
\begin{align}
\|\bm{f}^*_{r}\|_{r} &\le 4 \Big( \delta_{r,0} + 
\sum_{0< s\le r} c^s \|\bm{f}^*_{r-s}\|_{r-s}\Big)
\end{align}
and then can obtain inductively 
$\|\bm{f}^*_{r}\|_{r} \!\le\! 4(5c)^r$. 

By \eqref{expansion_of_Z}, we have
\begin{align}
\frac{\delta }{\delta J_{w_1}}Z_{\Lambda} &= \frac{\delta }{\delta J_{w_1}}\sum_{n\ge 1} \Big( \prod_{j=1}^{n-1} \int_{(B_{x_{1},\dots,x_{j}})^{\mathsf{c}}}\!\!\!dz_{j+1} \Big)\bm{Z}_{\Lambda;x_{1},\dots,x_{n}} Z_{\Lambda\backslash B_{x_{1},\dots,x_{n}}}\Big|_{x_1=w_1}\notag\\
&= \sum_{n\ge 1} \Big( \prod_{j=1}^{n-1} \int_{(B_{x_{1},\dots,x_{j}})^{\mathsf{c}}}\!\!\!dz_{j+1} \Big)Z_{\Lambda\backslash B_{x_{1},\dots,x_{n}}}\frac{\delta }{\delta J_{w_1}}\bm{Z}_{\Lambda;x_{1},\dots,x_{n}} \Big|_{x_1=w_1},
\end{align}
where we have used the fact $\tfrac{\delta }{\delta J_{w_1}}Z_{\Lambda\backslash B_{x_{1},\dots,x_{n}}}\big|_{x_1=w_1}\!\equiv\!0$.
Then, for $w_1,\dots,w_r \!\in\! \mathbb{R}^d$,
\begin{align}
S^c_{r,\Lambda;w_1,\dots,w_r} &= \frac{\delta^{r-1} }{\delta J_{w_2}\!\cdots\delta J_{w_r}} 
\Big( \frac{1}{Z_{\Lambda}} \frac{\delta }{\delta J_{w_1}} Z_{\Lambda} \Big) 
\Big|_{J=0} \notag\\
&= \sum_{n\ge 1} \Big( \prod_{j=1}^{n-1} \int_{(B_{x_1,\dots,x_{j}})^{\mathsf{c}}}\!\!\!dx_{j+1} \Big)
\frac{\delta^{r-1} }{\delta J_{w_2}\!\cdots\delta J_{w_r}} \notag\\
&\quad\, \Big( \bm{f}_{\Lambda;x_1,\dots,x_n} \frac{\delta}{\delta J_{w_1}} \bm{Z}_{\Lambda;x_1,\dots,x_n} \Big)
\Big|_{J=0,x_1=w_1} \notag\\
&= \sum\nolimits_{\substack{I,I'\subset\{1,\dots,r\}\\1\in I,I'=\{1,\dots,r\}\backslash I}}
\sum_{n\ge 1} \Big( \prod_{j=1}^{n-1} \int_{(B_{x_1,\dots,x_{j}})^{\mathsf{c}}}\!\!\!dx_{j+1} \Big) \notag\\
&\quad\; \bm{Z}_{|I|,\Lambda;(w_i)_{i\in I};x_1,\dots,x_n} \bm{f}_{|I'|,\Lambda;(w_i)_{i\in I'};x_1,\dots,x_n}
\Big|_{x_1=w_1},
\end{align}
which can be rewritten as $S^c_{r,\Lambda} = \sum_{1\le s\le r} \bm{T}_{s,\Lambda} \bm{f}_{r-s,\Lambda}$ with
\begin{gather}
( \bm{T}_{s,\Lambda} \bm{f} )_{w_1,\dots,w_r} = 
\sum\nolimits_{\substack{I,I'\subset\{1,\dots,r\}\\1\in I,|I|=s,I'=\{1,\dots,r\}\backslash I}} \sum_{n\ge 1} \notag\\
\Big( \prod_{j=1}^{n-1} \int_{(B_{x_1,\dots,x_{j}})^{\mathsf{c}}}\!\!\!dx_{j+1} \Big) 
\bm{Z}_{s,\Lambda;(w_i)_{i\in I};x_1,\dots,x_n} \bm{f}_{(w_i)_{i\in I'};x_1,\dots,x_n} \Big|_{x_1=w_1}.
\end{gather}
The formal limit of $S^c_{r,\Lambda}$ as $\Lambda \!\to\! \mathbb{R}^d$ is 
$S^c_{r} = \sum_{1\le s\le r} \bm{T}_{s} \bm{f}^*_{r-s}$ with
\begin{gather}
( \bm{T}_{s} \bm{f} )_{w_1,\dots,w_r} = 
\sum\nolimits_{\substack{I,I'\subset\{1,\dots,r\}\\1\in I,|I|=s,I'=\{1,\dots,r\}\backslash I}} \sum_{n\ge 1} \notag\\
\Big( \prod_{j=1}^{n-1} \int_{(B_{x_1,\dots,x_{j}})^{\mathsf{c}}}\!\!\!dx_{j+1} \Big)
\bm{Z}_{s,\mathbb{R}^d;(w_i)_{i\in I};x_1,\dots,x_n} \bm{f}_{(w_i)_{i\in I'};x_1,\dots,x_n} \Big|_{x_1=w_1}.
\end{gather}
For $r \!\ge\! 1$, let $\mathcal{F}'_{r}$ be the Banach space of 
Borel measurable functions $f\!: \mathbb{R}^{d\times r} \!\to\! \mathbb{R}$ with the norm 
\begin{align}
\|f\|'_{r} = \tfrac{1}{r!} \esssup_{\bm{w}\in \mathbb{R}^{d\times r}} 
e^{ \frac12 \ell_{\bm{w}}} |f_{\bm{w}}|,
\end{align}
where $\ell_{\bm{w}}$ is used to show the exponential decay of $\bm{f}_{\bm{w}}$ in $\bm{w}$.
Also let
\begin{align}
\|\bm{T}\|'_{s\to r} = \sup_{\|\bm{f}\|_{s}=1}
\| \bm{T} \bm{f} \|'_{r}.
\end{align}
for linear  operator $\bm{T}\!: \mathcal{F}_{s} \!\to\! \mathcal{F}'_{r}$.

\begin{theorem}
For $\lambda \!>\! 0$ sufficiently small and $1 \!\le\! s \!\le\! r$, 
$\|\bm{T}_{s}\|'_{r-s\to r} \le c^s$.
\end{theorem}

By Theorem 2, we have, for $\lambda \!>\! 0$ sufficiently small,
\begin{align}
\|S^c_{r}\|'_{r} &\le \sum_{1\le s\le r} 
\|\bm{T}_{s}\|'_{r-s\to r} \|\bm{f}_{r-s}^*\|_{r-s}
\le \sum_{1\le s\le r} c^s \cdot 4(5c)^{r-s} \le (5c)^r,
\end{align}
which is equivalent to
\begin{align}
|S^c_{r;w_1,\dots,w_r}| \le  (5c)^r r!\, e^{-\ell_{w_1,\dots,w_r}},\; w_1,\dots,w_r\!\in\!\mathbb{R}^d.
\end{align}

\section{Estimation}

Let $\mathcal{H}_n$ be the set of all sequences $(h_j)_{j=1}^n$ with $h_j\!\in\! \mathbb{N}$. 
For $(h_j)_{j=1}^n,(h'_j)_{j=1}^n \!\in\! \mathcal{H}_n$, we write $(h_j)_{j=1}^n\!\le\!(h'_j)_{j=1}^n$ iff $h_j\!\le\!h'_j$ for $1\!\le\! j\!\le\! n$ and $(h_j)_{j=1}^n\!\pm\!(h'_j)_{j=1}^n\!:=\!(h_j\!\pm\!h'_j)_{j=1}^n$. 
We define, for $\bm{x}\!=\!(x_j)_{j=1}^n\!\in\!\mathcal{X}_{n}$ and $\bm{h}\!=\!(h_j)_{j=1}^n\!\in\! \mathcal{H}_n$,
\begin{align}
\|F\|_{\bm{x};\bm{h}}^\phi=
\Big(\prod_{j=1}^n\frac{1}{h_j!}\prod_{k=1}^{h_j}\int_{B'_{x_1,\dots,x_j}}\!\!\!dy_{j,k}\Big)\bigg|\Big(\prod_{j=1}^n\prod_{k=1}^{h_j}\frac{\delta}{\delta\phi_{y_{j,k}}}\Big) F\bigg|.
\end{align}
Then we have
\begin{align}\label{phi_norm_1}
&\quad\,\,\|F_1F_2\|_{\bm{x};\bm{h}}^\phi=
\Big(\prod_{j=1}^n\frac{1}{h_j!}\prod_{k=1}^{h_j}\int_{B'_{x_1,\dots,x_j}}\!\!\!dy_{j,k}\Big)\bigg|\Big(\prod_{j=1}^n\prod_{k=1}^{h_j}\frac{\delta}{\delta\phi_{y_{j,k}}}\Big) F_1 F_2\bigg|\notag\\
&\le\sum_{\substack{h'_1,\dots,h'_n\in\mathbb{N}\\0\le h'_j\le h_j}}\!\!
\Big(\prod_{j=1}^n\frac{1}{h'_j!(h_j\!-\!h'_j)!}\prod_{k=1}^{h_j}\int_{B'_{x_1,\dots,x_j}}\!\!\!dy_{j,k}\Big)\bigg|\Big(\prod_{j=1}^n\prod_{k=1}^{h'_j}\frac{\delta}{\delta\phi_{y_{j,k}}}\Big) F_1 \bigg|\notag\\
&\quad\,\,\bigg|\Big(\prod_{j=1}^n\prod_{k=h'_j+1}^{h_j}\frac{\delta}{\delta\phi_{y_{j,k}}}\Big)  F_2\bigg|
=\sum_{\bm{h}'\in \mathcal{H}_n:\,\bm{h}'\le\bm{h}}\|F_1\|_{\bm{x};\bm{h}'}^\phi\|F_2\|_{\bm{x};\bm{h}-\bm{h}'}^\phi
\end{align}
and in general
\begin{align}\label{phi_norm_2}
\|F_1\cdots F_r\|_{\bm{x};\bm{h}}^\phi\le
\sum_{\substack{\bm{h}_1,\dots,\bm{h}_r\in \mathcal{H}_n\\\bm{h}_1+\cdots+\bm{h}_r=\bm{h}}}\|F_1\|_{\bm{x};\bm{h}_1}^\phi\!\cdots\|F_r\|_{\bm{x};\bm{h}_r}^\phi.
\end{align}
In order to deal with some singular functional conveniently in the following, 
we regard them as signed measures on $\mathbb{R}^d$. 
(For signed measure $\nu$, we have the Jordan decomposition $\nu\!=\nu^{+}\!-\!\nu^{-}$, 
where $\nu^{+}$, $\nu^{-}$ are the positive and negative variations of $\nu$, 
and we have the total variation of $\nu$ as $|\nu|\!=\!\nu^{+}\!+\!\nu^{-}$. 
Furthermore, for signed measures $\nu_1$, $\nu_2$, we write $\nu_1 \!\ge\! \nu_2$ iff $\nu_1 \!-\! \nu_2$ is a positive measure.
In particular, we have $|\delta_{z}(x)|\!=\! \delta_{z}(x)\!\ge\!0$.)
 
Then it is easy to show that
$\sum_{\bm{h}\in \mathcal{H}_n}\!\|\phi_{z}\|_{\bm{x};\bm{h}}^\phi\!\le\!
1\!+\!|\phi_z|$ and, by \eqref{phi_norm_2},
\begin{align}\label{phi_norm_of_pruduct}
\sum_{\bm{h}\in \mathcal{H}_n}\!\Big\|\prod_{k=1}^r\phi_{z_k}\Big\|_{\bm{x};\bm{h}}^\phi
\le\sum_{\bm{h}\in \mathcal{H}_n}\sum_{\substack{\bm{h}_1,\dots,\bm{h}_r\in\mathcal{H}_n\\\bm{h}_1+\cdots+\bm{h}_r=\bm{h}}}\prod_{k=1}^r\|\phi_{z_k}\|_{\bm{x};\bm{h}_k}^\phi\notag\\
=\prod_{k=1}^r\sum_{\bm{h}_k\in \mathcal{H}_n}\!\!\|\phi_{z_k}\|_{\bm{x};\bm{h}_k}^\phi
\le\prod_{k=1}^r\big(1\!+\!|\phi_{z_k}|\big).
\end{align}
Also we have the following bound:

\begin{lemma}
For $\bm{x}\!\in\!\mathcal{X}_{n}$, $\bm{h}\!\in\! \mathcal{H}_n$ and $\lambda\!\ge\!0$,
\begin{align}
\Big\|e^{-\lambda\!\int_\Lambda\!dx\,\phi^4_x}\Big\|_{\bm{x};\bm{h}}^\phi\le e^{c_2\lambda n}.
\end{align}
\end{lemma}
\begin{proof}
For any $D\!\in\!\mathcal{B}'$, we have
\begin{align}
&\Big(\prod_{k=1}^h\int_D\!dy_k\Big)\bigg|\Big(\prod_{k=1}^h\frac{\delta}{\delta\phi_{y_k}}\Big) e^{-\lambda\!\int_D\!dx\,\phi^4_x}\bigg|\le e^{-\lambda\!\int_D\!dx\,\phi^4_x}\sum_{p_1,\dots,p_4\ge 0}^{\sum_{r=1}^4\! rp_r=h}\notag\\
&\frac{h!}{\prod_{r=1}^4 p_r!(r!)^{p_r}}\prod_{r=1}^4\bigg(\Big(\prod_{k=1}^{r}\int_D\!dy_k\Big)\bigg|\Big(\prod_{k=1}^{r}\frac{\delta}{\delta\phi_{y_k}}\Big)\lambda\! \int_D\!dx\,\phi^4_x\bigg|\bigg)^{p_r}\notag\\
&\le h!\,e^{-\lambda\!\int_D\!dx\,\phi^4_x}\sum_{p_1,\dots,p_4\ge 0}
\prod_{r=1}^4\frac{1}{p_r!}\Big(\frac{4!}{r!(4\!-\!r)!}\lambda\!\int_D\!dx\,|\phi_x|^{4-r}\Big)^{p_r}\notag\\
&= h!\exp\Big\{\lambda\!\int_D\!dx\,\big((1\!+\!|\phi_x|)^4-2\phi^4_x\big)\Big\}\le h!\, e^{c\lambda |D|}
\end{align}
with $c\!=\!\sup_{\xi\ge 0}\big((1\!+\!\xi)^4\!-\!2\xi^4\big)$.
Then, for $\bm{x}\!=\!(x_j)_{j=1}^n\!\in\!\mathcal{X}_{n}$ and $\bm{h}\!=\!(h_j)_{j=1}^n\!\in\! \mathcal{H}_n$,
\begin{align}
\Big\|e^{-\lambda\!\int_\Lambda\!dx\,\phi^4_x}\Big\|_{\bm{x};\bm{h}}^\phi
&=e^{-\lambda\!\int_{\Lambda\backslash B_{x_1,\dots,x_n}}\!\!dx\,\phi^4_x}\prod_{j=1}^n\frac{1}{h_j!}\Big(\prod_{k=1}^{h_j}\int_{\Lambda\cap B'_{x_1,\dots,x_j}}\!\!\!dy_{j,k}\Big)\notag\\
&\quad\,\,\bigg|\Big(\prod_{k=1}^{h_j}\frac{\delta}{\delta\phi_{y_{j,k}}}\Big) e^{-\lambda\!\int_{\Lambda\cap B'_{x_1,\dots,x_j}}\!\!dx\,\phi^4_x}\bigg|\le e^{c v_d\lambda n},
\end{align}
where we use the facts that the range of the integration over $y_{j,k}$ can be replaced by $\Lambda\!\cap\! B'_{x_1,\dots,x_j}$ here
and $|\Lambda\!\cap\! B'_{x_1,\dots,x_j}|$ can be bounded by $v_d$, the volume of the unit ball in $d$ dimensions.
\end{proof}

For $(x_j)_{j=1}^n\!\in\!\mathcal{X}_{n}$ and $\eta\!\in\!\mathcal{T}_n$, we have
\begin{align}
&\quad\,\,\bigg|\Big(\prod_{j=2}^n\Delta_{x_j;x_1,\dots,x_{\eta(j)}}^\phi\Big) e^{-\lambda\!\int_{\Lambda}\!dx\,\phi^4_x}F\bigg|\notag\\
&\le\Big(\prod_{j=1}^{n-1}d_{\eta}(j)!\Big)\Big(\prod_{j=2}^n \sup\nolimits_{y\in B_{x_{\eta(j)}}}\!\!C_{x_j,y}\Big)\bigg\|\Big(\prod_{j=2}^n\frac{\delta}{\delta\phi_{x_j}}\Big)e^{-\lambda\!\int_{\Lambda}\!dx\,\phi^4_x}F\bigg\|_{(x_j)_{j=1}^{n-1};\bm{d}_{\eta}}^\phi
\end{align}
with $d_{\eta}(j)\!=\!\big|\eta^{-1}(\{j\})\big|$ and $\bm{d}_{\eta}\!=\!(d_{\eta}(j))_{j=1}^{n-1}$.
Then, by \eqref{phi_norm_1}, Lemma 1 and the fact $\sup\nolimits_{y\in B_{x_{\eta(j)}}}\!\!C_{x_j,y}\le c_1 e^{2-2|x_j-x_{\eta(j)}|}$, we can continue with 
\begin{align}\label{bound2_exp_F}
&\le\Big(\prod_{j=1}^{n-1}d_{\eta}(j)!\Big)\Big(\prod_{j=2}^n c_1 e^{2-2|x_j-x_{\eta(j)}|}\Big)\sum_{\substack{\bm{h}\in\mathcal{H}_{n-1}\\\bm{h}\le\bm{d}_{\eta}}}\Big\|e^{-\lambda\!\int_{\Lambda}\!dx\,\phi^4_x}\Big\|_{(x_j)_{j=1}^{n-1};\bm{d}_{\eta}-\bm{h}}^\phi\notag\\
&\quad\;\,\bigg\|\!\sum\nolimits_{\substack{L,L'\subset\{2,\dots,n\}\\L'=\{2,\dots,n\}\backslash L}}(4\lambda)^{|L'|}\Big(\prod_{l\in L'}\phi^3_{x_l}\Big)\Big(\prod_{l\in L}\frac{\delta}{\delta\phi_{x_l}}\Big)F\bigg\|_{(x_j)_{j=1}^{n-1};\bm{h}}^\phi\notag\\
&\le\big(c_1e^{2+c_2\lambda}\big)^{n-1}\Big(\prod_{j=1}^{n-1}d_{\eta}(j)!\Big)e^{-2\sum_{j=2}^{n}|x_j-x_{\eta(j)}|}\,U_{x_1,\dots,x_n}^\phi(F),
\end{align}
where $U_{x_1,\dots,x_n}^\phi(F)$ is an abbreviation for
\begin{align}
\sum\nolimits_{\substack{L,L'\subset\{2,\dots,n\}\\L'=\{2,\dots,n\}\backslash L}}(4\lambda)^{|L'|}\sum_{\bm{h}\in\mathcal{H}_{n-1}}\bigg\|\Big(\prod_{l\in L'}\phi^3_{x_l}\Big)\Big(\prod_{l\in L}\frac{\delta}{\delta\phi_{x_l}}\Big)F\bigg\|_{(x_j)_{j=1}^{n-1};\bm{h}}^\phi.
\end{align}

We now deal with the three factors $U_{x_1,\dots,x_n}^\phi(F)$, $e^{-2\sum_{j=2}^{n}|x_j-x_{\eta(j)}|}$ and 
$\prod_{j=1}^{n-1}d_{\eta}(j)!$ in \eqref{bound2_exp_F} one by one. 
First, let us consider $U_{x_1,\dots,x_n}^\phi(F)$ with $F\!=\!\phi_{w_1}\cdots\phi_{w_r}$:
\begin{align}\label{expectation_of_U_phi}
U_{x_1,\dots,x_n}^\phi(\phi_{w_1}\cdots\phi_{w_r})
&\le\sum\nolimits_{\substack{L,L'\subset\{2,\dots,n\}\\L'=\{2,\dots,n\}\backslash L}}\sum\nolimits_{\substack{\tau:L\to\{1,\dots,r\}\\\tau \text{ injective}}}(4\lambda)^{|L'|}\Big(\prod_{l\in L}\delta_{x_l,w_{\tau(l)}}\Big)\notag\\
&\quad\,\sum_{\bm{h}\in\mathcal{H}_{n-1}}\bigg\|\Big(\prod_{l\in L'}\phi^3_{x_l}\Big)\Big(\prod_{i\in\{1,\dots,r\}\backslash\tau(L)}\phi_{w_i}\Big)\bigg\|_{(x_j)_{j=1}^{n-1};\bm{h}}^\phi\notag\\
&\le\sum\nolimits_{\substack{L,L'\subset\{2,\dots,n\}\\L'=\{2,\dots,n\}\backslash L}}\sum\nolimits_{\substack{\tau:L\to\{1,\dots,r\}\\\tau \text{ injective}}}(4\lambda)^{|L'|}\Big(\prod_{l\in L}\delta_{x_l,w_{\tau(l)}}\Big)\notag\\
&\quad\,
\prod_{l\in L'}\big(1\!+\!|\phi_{x_l}|\big)^3\prod_{i\in\{1,\dots,r\}\backslash\tau(L)}\big(1\!+\!|\phi_{w_i}|\big).
\end{align}
Before going further, we provide the following lemma, which is needed to bound the expectations $\langle U_{x_1,\dots,x_n}^\phi(\phi_{w_1}\cdots\phi_{w_r})\rangle_{\bm{t};\bm{x}}^\phi$.
\begin{lemma}
For $w_1,\dots,w_r\!\in\!\mathbb{R}^d$ and $(x_1,\dots,x_n)\!\in\!\mathcal{X}_{n}$,
\begin{align}
\Big\langle\prod_{i=1}^r \big(1\!+\!|\phi_{w_i}|\big)\prod_{j=1}^n \big(1\!+\!|\phi_{x_j}|\big)^3\Big\rangle_{\bm{t};\bm{x}}^\phi
\le c_3^r c_4^n (r!)^{1/2}.
\end{align}
\end{lemma}
\begin{proof}
First, we claim that, for $\sum_{j=1}^n\! s_j\!=\!2s$ even and  $(x_1,\dots,x_n)\!\in\!\mathcal{X}_{n}$,
\begin{align}
\bigg|\Big\langle\prod_{j=1}^n \phi_{x_j}^{s_j}\Big\rangle_{\bm{t};\bm{x}}^\phi\bigg|\le c^{s}\prod_{j=1}^n (s_j!)^{1/2},
\end{align}
which is proved by induction on $s$ and
is trivial for $s\!=\!0$. Assuming it holds for $s\!=\!\bar{s}\!-\!1$, we consider the case  $s\!=\!\bar{s}$ and assume further $s_1\!=\!\max_{j=1}^n\!s_j$ without loss of generality. By performing integration by parts and applying the inductive assumption, 
we have
\begin{align}
\bigg|\Big\langle\prod_{j=1}^n \phi_{x_j}^{s_j}\Big\rangle_{\bm{t};\bm{x}}^\phi\bigg|
&= \bigg|\Big\langle\!\int_{\mathbb{R}^d}\!\!dx\,(C_{\bm{t};\bm{x}})_{x_1,x}\frac{\delta}{\delta \phi_x}\Big(\phi_{x_1}^{s_1-1}\prod_{j=2}^n \phi_{x_j}^{s_j}\Big)\Big\rangle_{\bm{t};\bm{x}}^\phi\bigg|\notag\\
&\le (s_1\!-\!1)(C_{\bm{t};\bm{x}})_{x_1,x_1}\bigg|\Big\langle\phi_{x_1}^{s_1-2}\prod_{j=2}^n \phi_{x_j}^{s_j}\Big\rangle_{\bm{t};\bm{x}}^\phi\bigg|+\sum_{k=2}^n\notag\\
&\quad\,\, s_k (C_{\bm{t};\bm{x}})_{x_1,x_k}\bigg|\Big\langle\phi_{x_1}^{s_1-1}\phi_{x_k}^{s_k-1}
\prod_{\substack{j=2,\dots,n\\j\neq k}}\phi_{x_j}^{s_j} \Big\rangle_{\bm{t};\bm{x}}^\phi\bigg|\notag\\
&\le c^{\bar{s}-1}\prod_{j=1}^n (s_j!)^{1/2}\sum_{j=1}^n (C_{\bm{t};\bm{x}})_{x_1,x_j},
\end{align}
where $\phi_{x}^{-1}$ is temporarily defined as 0. 
Since
\begin{align}
&\sum_{j=1}^n (C_{\bm{t};\bm{x}})_{x_1,x_j}\le c_1\sum_{j=1}^n e^{-2|x_1-x_j|}\le \frac{2^dc_1e}{v_d}\sum_{j=1}^n\notag\\&\int_{|x-x_j|\le 1/2} \!\!dx\,e^{-2|x_1-x|}
\le\frac{2^dc_1e}{v_d} \int_{\mathbb{R}^d}\!\!dx\,e^{-2|x|}<\infty,
\end{align}
we can advance the induction for $c$ sufficiently large and thus complete the proof of the claim. 

Then, by the Cauchy-Schwartz inequality,
\begin{align}
\Big\langle\prod_{i=1}^r |\phi_{w_i}|\prod_{j=1}^n |\phi_{x_j}|^{r_j}\Big\rangle_{\bm{t};\bm{x}}^\phi
&\le\bigg(\Big\langle\prod_{i=1}^r \phi_{w_i}^2\Big\rangle_{\bm{t};\bm{x}}^\phi\Big\langle\prod_{j=1}^n \phi_{x_j}^{2r_j}\Big\rangle_{\bm{t};\bm{x}}^\phi\bigg)^{1/2}\notag\\
&\le \Big((2c_1)^r r! \prod_{j=1}^n c^{r_j}(2r_j)!^{1/2}\Big)^{1/2},
\end{align}
which yields
\begin{align}
&\quad\,\Big\langle\prod_{i=1}^r \big(1\!+\!|\phi_{w_i}|\big)\prod_{j=1}^n \big(1\!+\!|\phi_{x_j}|\big)^3\Big\rangle_{\bm{t};\bm{x}}^\phi\notag\\
&\le 2^{r+3n}\!\sup_{I\subset\{1,\dots,r\}}\sup_{\substack{r_1,\dots,r_n\\0\le r_j\le 3}}
\Big\langle\prod_{i\in I} |\phi_{w_i}|\prod_{j=1}^n |\phi_{x_j}|^{r_j}\Big\rangle_{\bm{t};\bm{x}}^\phi\notag\\
&\le 2^{r+3n}\big((2c_1)^{r}(\sqrt{6!}\,c^3)^{n} r!  \big)^{1/2}.
\end{align}
\end{proof}

The factor $e^{-2\sum_{j=2}^{n}|x_j-x_{\eta(j)}|}$ will play two roles in our derivation. 
One half of the factor is used later to extract some exponential decay in external points $w_1,\dots,w_r$. 
The other half is used to control the integrations over $x_2,\dots,x_n$ as follows.
By \eqref{expectation_of_U_phi} and Lemma 2, we have, for $\bm{t}\!\in\!\mathcal{I}_{n}$ and $\eta\!\in\!\mathcal{T}_n$,
\begin{align} \label{integration_of_U}
&\quad\,\,\Big(\prod_{j=1}^{n-1}\int_{(B_{x_{1},\dots,x_{j}})^{\mathsf{c}}}\!\!\!dx_{j+1}\Big)e^{-\sum_{j=2}^{n}|x_j-x_{\eta(j)}|}
\big\langle U^\phi_{x_1,\dots,x_n}(\phi_{w_1}\!\cdots\phi_{w_r})
\big\rangle_{\bm{t};(x_j)_{j=1}^{n-1}}^\phi\notag\\
&\le 
\Big(\prod_{j=1}^{n-1}\int_{(B_{x_{1},\dots,x_{j}})^{\mathsf{c}}}\!\!\!dx_{j+1}\Big)
e^{-\sum_{j=2}^{n}|x_j-x_{\eta(j)}|}\sum\nolimits_{\substack{L,L'\subset\{2,\dots,n\}\\L'=\{2,\dots,n\}\backslash L}}\sum\nolimits_{\substack{\tau:L\to\{1,\dots,r\}\\\tau \text{ injective}}}\notag\\
&\quad\;\,(4\lambda)^{|L'|}\Big(\prod_{l\in L}\delta_{x_l,w_{\tau(l)}}\Big)
\Big\langle\prod_{l\in L'}\big(1\!+\!|\phi_{x_l}|\big)^3\prod_{i\in\{1,\dots,r\}\backslash\tau(L)}\big(1\!+\!|\phi_{w_i}|\big)\Big\rangle_{\bm{t};(x_j)_{j=1}^{n-1}}^\phi\notag\\
&\le \sum\nolimits_{\substack{L\subset\{2,\dots,n\}\\|L|\le r}}\sum\nolimits_{\substack{\tau:L\to\{1,\dots,r\}\\\tau \text{ injective}}}(4c_4c_5\lambda)^{n-1-|L|}c_3^{r-|L|}(r\!-\!|L|)!^{1/2}\notag\\
&\le c_3^r r!(8c_4c_5\lambda)^{n-1}\max\{(4c_3c_4c_5\lambda)^{-\min\{r,n-1\}}, 1\}
\end{align}
with $c_5=\!\int_{\mathbb{R}^d}\!dx\, e^{-|x|}$. 
In the last line of \eqref{integration_of_U}, we use the facts that the number of subsets of $\{2,\dots,n\}$ is $2^{n-1}$ and the number of injective maps from $L$ to $\{1,\dots,r\}$ with $|L|\!\le\!r$ is $\frac{r!}{(r-|L|)!}$.
We assume from now on that $\lambda\!<\!(4c_3c_4c_5)^{-1}$.

To control the factor $\prod_{j=1}^{n-1}d_{\eta}(j)!$, we need the following lemma:
\begin{lemma} [Speer \cite{B99}] \label{lemma3} 
For $n\!\ge\!2$, 
\begin{align}\label{combinatorial_inequality}
\sum_{\eta\in\mathcal{T}_n}\!\Big(\prod_{j=1}^{n-1}d_{\eta}(j)!\Big)\Big(\prod_{j=1}^{n-1}\!\int_0^1\!\!dt_j\Big)\prod_{j=2}^n \prod_{l=\eta(j)}^{j-2}t_l=\frac{1}{n}\binom{2n\!-\!2}{n\!-\!1}\le 4^{n-1}.
\end{align} 
\end{lemma}
\begin{proof}
For completeness, we give a new proof for this lemma. Let
\begin{align}\label{g_x_1}
\omega=x(1\!-\!\omega)^{-1}=\frac{1\!-\!\sqrt{1\!-\!4x}}{2}=\sum_{n=1}^\infty \frac{1}{n}\binom{2n\!-\!2}{n\!-\!1} x^n.
\end{align}
Given $\eta\!\in\!\mathcal{T}_n$, we have
\begin{align}\label{eq1_of_lemma3}
&\quad\,\prod_{j=1}^n{\Big(1-\omega\prod_{l=j}^{n-1}t_l \Big)^{-d_{\eta}(j)-1}}-1=\int_0^1\!\!dt_n\,\frac{\partial}{\partial t_n}\prod_{j=1}^n{\Big(1-\omega\prod_{l=j}^{n}t_l \Big)^{-d_{\eta}(j)-1}}\notag\\
&=\sum_{k=1}^n\big(d_{\eta}(k)\!+\!1\big)\int_0^1\!\!dt_n\Big(\prod_{l=k}^{n-1}t_l\Big)x{(1\!-\!\omega)^{-1}}\prod_{j=1}^n{\Big(1-\omega\prod_{l=j}^{n}t_l \Big)^{-d_{\eta}(j)-\delta_{j,k}-1}},
\end{align}
where $d_{\eta}(n)$ is defined as 0. Multiplying both sides of \eqref{eq1_of_lemma3} by 
\begin{align}
\Big(\prod_{j=1}^{n-1}d_{\eta}(j)!\Big)\Big(\prod_{j=2}^n \prod_{l=\eta(j)}^{j-2}t_l\Big),
\end{align}
performing integrations over $t_1,\dots,t_{n-1}$ and summing over all $\eta\!\in\!\mathcal{T}_n$, we obtain $A_n(x) \!-\! B_n \!=\! xA_{n+1}(x)$ with
\begin{gather} 
A_n(x) = \sum_{\eta\in\mathcal{T}_n}\!\Big(\prod_{j=1}^{n-1}d_{\eta}(j)!\Big)\Big(\prod_{j=1}^{n-1}\!\int_0^1\!\!dt_j\Big)\Big(\prod_{j=2}^n \prod_{l=\eta(j)}^{j-2}t_l\Big)\prod_{j=1}^n{\Big(1-\omega\prod_{l=j}^{n-1}t_l \Big)^{-d_{\eta}(j)-1}}\\
\intertext{and}
B_n = \sum_{\eta\in\mathcal{T}_n}\!\Big(\prod_{j=1}^{n-1}d_{\eta}(j)!\Big)\Big(\prod_{j=1}^{n-1}\!\int_0^1\!\!dt_j\Big)\prod_{j=2}^n \prod_{l=\eta(j)}^{j-2}t_l,
\end{gather}
which yields
\begin{align}\label{g_x_2}
\omega=x(1\!-\!\omega)^{-1}&=x+x^2\!\int_0^1\!\!dt_1(1\!-\!t_1 \omega)^{-2}{(1\!-\!\omega)^{-1}}=x+x^2 A_{2}(x)\notag\\
&=x+x^2 B_{2} +x^3 A_3(x)=\cdots=x+\sum_{n\ge 2}x^n B_n.
\end{align}
Comparing \eqref{g_x_1} and \eqref{g_x_2}, we complete the proof.
\end{proof}

Now we are ready to perform the estimation for $\bm{Z}_{r,\Lambda;w_1,\dots,w_r;x_1,\dots,x_n}$ with $n\!\ge\!2$.
First, we have 
\begin{align}
&\,\big|\bm{Z}_{r,\Lambda;w_1,\dots,w_r;x_1,\dots,x_n}\big|\le\sum_{\eta\in\mathcal{T}_n}\!\Big(\prod_{j=1}^{n-1}\!\int_0^1\!\!dt_j\Big)\Big(\prod_{j=2}^n \prod_{l=\eta(j)}^{j-2}t_l\Big)\notag\\
&\bigg\langle\Big|\Big(\prod_{j=2}^n\Delta_{x_j,B'_{x_1,\dots,x_{\eta(j)}}}^\phi\Big) e^{-\lambda\!\int_{\Lambda\cap B_{x_1,\dots,x_n}}\!\!dx\,\phi^4_x}\phi_{w_1}\!\cdots\phi_{w_r}\Big|\bigg\rangle_{(t_j)_{j=1}^{n-1};(x_j)_{j=1}^{n-1}}^\phi \notag\\
&\!\le\sum_{\eta\in\mathcal{T}_n}\!\Big(\prod_{j=1}^{n-1}d_{\eta}(j)!\Big)\Big(\prod_{j=1}^{n-1}\!\int_0^1\!\!dt_j\Big)\Big(\prod_{j=2}^n \prod_{l=\eta(j)}^{j-2}t_l\Big)\big(c_1 e^{2+c_2\lambda}\big)^{n-1}\notag\\
&\,e^{-2\sum_{j=2}^{n}|x_j-x_{\eta(j)}|}\big\langle U^\phi_{x_1,\dots,x_n}(\phi_{w_1}\!\cdots\phi_{w_r})\big\rangle_{(t_j)_{j=1}^{n-1};(x_j)_{j=1}^{n-1}}^\phi.
\end{align}
Using \eqref{integration_of_U}, Lemma \ref{lemma3} and the fact $\sum_{j=2}^{n}|x_j\!-\!x_{\eta(j)}| \!\ge\! \ell_{x_{1},\dots,x_{n}}$, we obtain that
\begin{align}\label{integration_of_Z}
&\quad\;\Big(\prod_{j=1}^{n-1}\int_{(B_{x_{1},\dots,x_{j}})^{\mathsf{c}}}\!\!\!dx_{j+1}\Big)
\big|\bm{Z}_{r,\Lambda;w_1,\dots,w_r;x_1,\dots,x_n}\big|\,e^{\ell_{x_{1},\dots,x_{n}}}\notag\\
&\le\sum_{\eta\in\mathcal{T}_n}\!\Big(\prod_{j=1}^{n-1}d_{\eta}(j)!\Big)\Big(\prod_{j=1}^{n-1}\!\int_0^1\!\!dt_j\Big)\Big(\prod_{j=2}^n \prod_{l=\eta(j)}^{j-2}t_l\Big)\Big(\prod_{j=1}^{n-1}\int_{(B_{x_{1},\dots,x_{j}})^{\mathsf{c}}}\!\!\!dx_{j+1}\Big)\notag\\
&\quad\;\big(c_1e^{2+c_2\lambda}\big)^{n-1}e^{-\sum_{j=2}^{n}|x_j-x_{\eta(j)}|}\big\langle U_{x_1,\dots,x_n}^\phi(\phi_{w_1}\!\cdots\phi_{w_r})\big\rangle_{(t_j)_{j=1}^{n-1};(x_j)_{j=1}^{n-1}}^\phi\notag\\
&\le c_3^r r!\big(32c_1c_4c_5\lambda\, e^{2+c_2\lambda}\big)^{n-1}(4c_3c_4c_5\lambda)^{-\min\{r,n-1\}}.
\end{align}

\begin{proof}[Proof of Theorem 1]
For $\bm{w}\!\in\!\mathbb{R}^{dr}$ and $(\bm{x},z_1)\!\in\!\mathcal{X}_{n+1}$,
\begin{align}
&\big|(\bm{A}_{0}\bm{f})_{\bm{w};\bm{x},z_1}\big|\le\bm{1}_{\bm{x}\not\in\mathcal{X}_0}|\bm{f}_{\bm{w};\bm{x}}|+|\bm{Z}_{0,(B_{\bm{x}})^{\mathsf{c}};z_{1}}\!-\!1||\bm{f}_{\bm{w};\bm{x},z_1}|\;+\notag\\
&\!\sum_{m\ge 2}\Big(\prod_{j=1}^{m-1}\int_{(B_{\bm{x},z_{1},\dots,z_{j}})^{\mathsf{c}}}\!\!\!dz_{j+1}\Big)
|\bm{Z}_{0,(B_{\bm{x}})^{\mathsf{c}};z_{1},\dots,z_{m}}||\bm{f}_{\bm{w};\bm{x},z_{1},\dots,z_{m}}|\notag\\
&\!\le r!\|\bm{f}\|_{r}\Big(2^{n-1} e^{-\ell_{\bm{w};\bm{x}}}+2^{n}|\bm{Z}_{0,(B_{\bm{x}})^{\mathsf{c}};z_{1}}\!-\!1| e^{-\ell_{\bm{w};\bm{x},z_1}}+\sum_{m\ge 2}2^{n+m-1}\notag\\
&\Big(\prod_{j=1}^{m-1}\int_{(B_{\bm{x},z_{1},\dots,z_{j}})^{\mathsf{c}}}\!\!\!dz_{j+1}\Big)|\bm{Z}_{0,(B_{\bm{x}})^{\mathsf{c}};z_{1},\dots,z_{m}}|e^{-\ell_{\bm{w};\bm{x},z_1,\dots,z_m}}\Big).
\end{align}
By the inequalities
$\ell_{\bm{w};\bm{x}}\ge\ell_{\bm{w};\bm{x},z_1}$ and
$\ell_{\bm{w};\bm{x},z_1,\dots,z_m}\ge\ell_{\bm{w};\bm{x},z_1}-\ell_{z_1,\dots,z_m}$,
we have
\begin{align}
&\,\|\bm{A}_{0}\|_{r\to r}=\sup_{\|\bm{f}\|_{r}=1}\|\bm{A}_0\bm{f}\|_{r}\notag\\
&\le \esssup_{(\bm{x},z_1)\in \bigcup_{n\ge 1}\!\mathcal{X}_{n}}\!\Big(\tfrac12 +|\bm{Z}_{0,(B_{\bm{x}})^{\mathsf{c}};z_{1}}\!-\!1|+\!\sum_{m\ge 2}2^{m-1}\notag\\
&\,\Big(\prod_{j=1}^{m-1}\int_{(B_{\bm{x},z_{1},\dots,z_{j}})^{\mathsf{c}}}\!\!\!dz_{j+1}\Big)|\bm{Z}_{0,(B_{\bm{x}})^{\mathsf{c}};z_{1},\dots,z_{m}}|\,e^{\ell_{z_1,\dots,z_m}}\Big),
\end{align}
where
\begin{align}
\big|\bm{Z}_{0,(B_{\bm{x}})^{\mathsf{c}};z_{1}}\!-\!1\big|&=\lambda\!\int_0^1\!\!dt\!\int_{B'_{\bm{x},z_1}}\!\!\!dy\!\int\! d\mu(\phi)\, 
e^{-t\lambda\!\int_{B'_{\bm{x},z_1}}\!\!dx\,\phi^4_x}\phi^4_y\notag\\
&\le\lambda\!\int_{B_{z_1}}\!\!\!dy\!\int\! d\mu(\phi) 
\,\phi^4_{y}\le 3c_1^2 v_d\lambda
\end{align}
and the integration of $|\bm{Z}_{0,(B_{\bm{x}})^{\mathsf{c}};z_{1},\dots,z_{m}}|\,e^{\ell_{z_1,\dots,z_m}}$ is bounded by $\big(32c_1c_4c_5\lambda\, e^{2+c_2\lambda}\big)^{m-1}$.
Thus $\|\bm{A}_{0}\|_{r}\le \tfrac12 +3c_1^2 v_d\lambda+\sum_{m\ge 2}\big(64c_1c_4c_5\lambda\, e^{2+c_2\lambda}\big)^{m-1}\le\tfrac34$ for $\lambda$ sufficiently small.

Also, for $1\!\le\! s\!\le\! r$, $w_1,\dots,w_r\!\in\!\mathbb{R}^d$ and $(\bm{x},z_1)\!\in\!\mathcal{X}_{n+1}$,
\begin{align}
\big|(\bm{A}_{s}\bm{f})_{w_1,\dots,w_r;\bm{x},z_1}\big|&\le\sum\nolimits_{\substack{I,I'\subset\{1,\dots,r\}\\|I|=s,I'=\{1,\dots,r\}\backslash I}}\sum_{m\ge 1}
\Big(\prod_{j=1}^{m-1}\!\int_{(B_{\bm{x},z_{1},\dots,z_{j}})^{\mathsf{c}}}\!\!\!dz_{j+1}\Big)\notag\\
&\quad\;\big|\bm{Z}_{s,(B_{\bm{x}})^{\mathsf{c}};(w_i)_{i\in I};z_{1},\dots,z_{m}}\big|\big|\bm{f}_{(w_i)_{i\in I'};\bm{x},z_{1},\dots,z_{m}}\big|\notag\\
&\le(r\!-\!s)!\,\|\bm{f}\|_{r-s}\sum\nolimits_{\substack{I,I'\subset\{1,\dots,r\}\\|I|=s,I'=\{1,\dots,r\}\backslash I}}\sum_{m\ge 1}\,2^{n+m-1}\notag\\
&\quad\;\Big(\prod_{j=1}^{m-1}\!\int_{(B_{\bm{x},z_{1},\dots,z_{j}})^{\mathsf{c}}}\!\!\!dz_{j+1}\Big)
\big|\bm{Z}_{s,(B_{\bm{x}})^{\mathsf{c}};(w_i)_{i\in I};z_{1},\dots,z_{m}}\big|\notag\\
&\quad\;\exp\big\{\!-\!\ell_{(w_i)_{i\in I'};\bm{x},z_1,\dots,z_m}\!\big\}.
\end{align}
Since
$\bm{Z}_{s,(B_{\bm{x}})^{\mathsf{c}};(w_i)_{i\in I};z_{1},\dots,z_{m}}\!=\!0$ for $\{w_i|\,i\!\in\! I\}\!\not\subset\!B_{z_{1},\dots,z_{m}}$ and
\begin{align}
\ell_{(w_i)_{i\in I'};\bm{x},z_1,\dots,z_m}
&\ge\ell_{w_1,\dots,w_r;\bm{x},z_1,\dots,z_m} -\ell_{(w_i)_{i\in I};\bm{x},z_1,\dots,z_m}\notag\\
&\ge\ell_{w_1,\dots,w_r;\bm{x},z_1} -\ell_{z_1,\dots,z_m} - s
\end{align}
for $\{w_i|\,i\!\in\! I\}\!\subset\!B_{z_{1},\dots,z_{m}}$, we have
\begin{align}
&\|\bm{A}_{s}\|_{r-s\to r}=\sup_{\|\bm{f}\|_{r-s}=1}\|\bm{A}_s\bm{f}\|_{r}\notag\\
&\!\le \tfrac{(r-s)!}{r!}\esssup\nolimits_{\substack{w_1,\dots,w_r\in\mathbb{R}^d\\(\bm{x},z_1)\in \bigcup_{n\ge 1}\!\mathcal{X}_{n}}}\sum\nolimits_{\substack{I,I'\subset\{1,\dots,r\}\\|I|=s,I'=\{1,\dots,r\}\backslash I}}\sum_{m\ge 1}2^{m-1}\notag\\
&\Big(\prod_{j=1}^{m-1}\!\int_{(B_{\bm{x},z_{1},\dots,z_{j}})^{\mathsf{c}}}\!\!\!dz_{j+1}\Big)
\big|\bm{Z}_{s,(B_{\bm{x}})^{\mathsf{c}};(w_i)_{i\in I};z_{1},\dots,z_{m}}\big|\,e^{\ell_{z_1,\dots,z_m}+s}\notag\\
&\!\le (e c_3 )^s \sum_{m\ge 1}\big(64c_1c_4c_5\lambda\, e^{2+c_2\lambda}\big)^{m-1}(4c_3c_4c_5\lambda)^{-\min\{s,m-1\}}\notag\\
&\!\le \big(e c_3 (1\!+\!32e^2c_1c_3^{-1})\big)^s
\end{align}
for $\lambda$ sufficiently small.
\end{proof}

\begin{proof}[Proof of Theorem 2]
For $1\!\le\! s\!\le\! r$ and $w_1,\dots,w_r\!\in\!\mathbb{R}^d$,
\begin{align}
\big|(\bm{T}_{s,\Lambda}\bm{f})_{w_1,\dots,w_r}\big|
&\le(r\!-\!s)!\,\|\bm{f}\|_{r-s}\sum\nolimits_{\substack{I,I'\subset\{1,\dots,r\}\\1\in I,|I|=s,I'=\{1,\dots,r\}\backslash I}}\sum_{n\ge 1}\notag\\
&\,2^{n-1}\Big(\prod_{j=1}^{n-1}\int_{(B_{x_1,\dots,x_{j}})^{\mathsf{c}}}\!\!\!dx_{j+1}\Big)\big|\bm{Z}_{s,\Lambda,\mathbb{R}^d;(w_i)_{i\in I};x_1,\dots,x_n}\big|\notag\\
&\exp\big\{\!-\!\ell_{(w_i)_{i\in I'};x_1,\dots,x_n}\!\big\}\Big|_{x_1=w_1},
\end{align}
where
\begin{align}
\ell_{(w_i)_{i\in I'};x_1,\dots,x_n}
&\ge \ell_{w_1,\dots,w_r;x_1,\dots,x_n} - \ell_{(w_i)_{i\in I};x_1,\dots,x_n}\notag\\
&\ge \ell_{w_1,\dots,w_r} - \ell_{x_1,\dots,x_n}+ 1 - s
\end{align}
for $x_1\!=\!w_1$ and $\{w_i|\,i\!\in\! I\}\!\subset\!B_{x_1,\dots,x_n}$.
Then
\begin{align}
&\|\bm{T}_{s}\|'_{r-s\to r}
\le\tfrac{(r-s)!}{r!}\esssup_{w_1,\dots,w_r\in\mathbb{R}^d} \sum\nolimits_{\substack{I,I'\subset\{1,\dots,r\}\\1\in I,|I|=s,I'=\{1,\dots,r\}\backslash I}}\sum_{n\ge 1}\notag\\
&2^{n-1}\Big(\prod_{j=1}^{n-1}\int_{(B_{x_1,\dots,x_{j}})^{\mathsf{c}}}\!\!\!dx_{j+1}\Big)
\big|\bm{Z}_{s,\mathbb{R}^d;(w_i)_{i\in I};x_1,\dots,x_n}\big| \,e^{\ell_{x_1,\dots,x_n}+s}\Big|_{x_1=w_1}\notag\\
&\!\le (e c_3)^s
\sum_{n\ge 1}\big(64c_1c_4c_5\lambda\, e^{2+c_2\lambda}\big)^{n-1}(4c_3c_4c_5\lambda)^{-\min\{s,n-1\}}\notag\\
&\!\le \big(e c_3(1\!+\!32e^2c_1c_3^{-1})\big)^s
\end{align}
for $\lambda$ sufficiently small.
\end{proof}

\section*{Acknowledgements}
The work is partially supported by Wu Wen-Tsun Key Laboratory of Mathematics.
The author thanks Professor Zheng Yin for valuable advices and suggestions.

\end{document}